\documentclass[11pt]{article}

\usepackage{amssymb}
\usepackage{amsmath}
\usepackage{graphics}
\usepackage{color} 
\usepackage{fullpage}

\usepackage[lined,boxed,commentsnumbered]{algorithm2e}

\newcommand{\remove}[1]{}

\newtheorem {theorem}{Theorem}[section]
\newtheorem {lemma}[theorem]{Lemma}
\newtheorem {proposition}[theorem]{Proposition}
\newtheorem {corollary}[theorem]{Corollary}

\newtheorem {definition}[theorem]{Definition}

\newtheorem {example}[theorem]{Example}
\newtheorem {fact}[theorem]{Fact}

\newcommand{\ignore}[1]{}

\renewcommand\th{^{\text{th}}}

\newcommand\R{\mathbb{R}}

\def\fast {\mathit{fast}}
\def\part {\mathit{part}}

\def\view{\mathit{view}}

\def\Runs{\mathcal{R}}

\def\A {\mathbb{A}}

\def\s {\mathbf{s}}
\def\t {\mathbf{t}}

\def\Chr{\operatorname{Chr}}

\def\WF{\mathit{WF}}

\def\Res{\mathit{Res}}
\def\OF{\mathit{OF}}
\def\O {\mathcal{O}}
\def\I {\mathcal{I}}

\def\slow{\mathit{slow}}

\def\V{\mathcal{V}}

\def\adv{\operatorname{adv}}

\newenvironment{proof}[1][Proof]{\noindent\textbf{#1.} }{\hfill $\Box$\\[2mm]} 

\newfont{\mycrnotice}{ptmr8t at 7pt}
\newfont{\myconfname}{ptmri8t at 7pt}

\ignore{

\usepackage{amssymb}
\usepackage{amsmath}
\usepackage{graphics}
\usepackage{color} 
\usepackage{fullpage}

\usepackage{vmargin}
\setmarginsrb{28mm}{10mm}{20mm}{15mm}{12pt}{11mm}{0pt}{11mm}

\newtheorem{theorem}{Theorem}
 \newtheorem{definition}[theorem]{Definition}
 \newtheorem{lemma}[theorem]{Lemma}
 \newtheorem{proposition}[theorem]{Proposition}
\newenvironment{proof}[1][Proof]{\noindent\textbf{#1.} }{\hfill $\Box$\\[2mm]} 

\newcounter{linenumber}

\def\A{\ensuremath{\mathcal{A}}}

\def\R{\ensuremath{\mathcal{R}}}

\def\I{\ensuremath{\mathcal{I}}}
\def\O{\ensuremath{\mathcal{O}}}

\def\V{\ensuremath{\mathcal{V}}}

\newcommand{\remove}[1]{}

\newcommand{\ignore}[1]{}
}
\begin{document}

\bibliographystyle{plain}


\title{Set Consensus: Captured by a Set of Runs with Ramifications}

\author{
Eli Gafni$$
\\
\\
\\
\large $$ UCLA\\
eli@ucla.edu\\
}

\date{}

\maketitle

\begin{abstract}

\noindent Are (set)-consensus objects necessary? This paper answer is negative.

We show that the availability of consensus objects can be replaced by restricting the set of runs we consider.
In particular we concentrate of the set of runs of the Immediate-Snapshot-Model (IIS), and given the object we identify this restricted subset of IIS runs.

We further show that given an $(m,k)$-set consensus, an object that provides $k$-set consensus among $m$ processors, in a system of $n$, $n>m$ processors, we do not need to use
the precise power of the objects but rather their effective cumulative set consensus power. E.g. when $n=3,~m=2,$ and $k=1$ and all the 3 
processors are active then we only use 2-set consensus among the 3 processors, as if 2-processors consensus is not available. We do this until at least one of the 3 processors obtains an output.
We show that this suggests a new direction in the design of algorithms when consensus objects are involved.

\end{abstract}






\newpage

\section{Introduction}

We present 3 contributions:
\begin{enumerate}
\item Showing equivalence between $MK$, the SWMR system equipped with $(m,k)$-set consensus, and an IIS subset of runs $\cal{RMK}$,
\item Show how to compile programs written for $MK$ to be run on $\cal{RMK}$, and
\item Use the observation of how the this compiler works to design algorithms with $\cal{RMK}$ in mind, instead of the explicit objects.
The paper illustrates this possibility on two elementary examples. This open the possibility of a ``new design paradigm'' that reduces ``designing with objects'' with ``designing with reduced concurrency.''
\end{enumerate}

The first two contributions are theoretical in nature. We show their feasibility without worrying about the complexity of the various implementations. These contributions are small steps in a rather ambitious agenda, and it makes the agenda more credible. The agenda involves establishing the following: 

\begin{enumerate}
\item Any interesting distributed computing model can solve $\epsilon$-agreement (there is no interesting model below read-write wait-free),
\item Any interesting distributed problem, at the possibility level, can be captured by a network of  \emph{tasks} (tasks are the functions of distributed computing), 
\item Any interesting distributed system is equivalent to some subset of runs of the wait-free IIS model (read-write and the knowledge of the possible run captures any possible ``knowledge-relation'' among processors),
\item Any question about the solvability of a task $T$ in a (sub-IIS) model $M$ can be reduced to the question of solvability of a task $T(M)$ wait-free (It is all wait-free).
\end{enumerate}

The contribution to the agenda is in showing that indeed the $MK$ model (a very interesting distributed system), corresponds to a subset of runs of the IIS.

The second less theoretical contribution is in a ``new algorithm design paradigm'' that emerges from the theoretical results: Any task solvable with set consensus can be solved by a ``wait-free'' algorithm tuned to the variable concurrency provided by the objects.

There are three main technical ideas involved:

\begin{enumerate}

\item Describe the availability of $(m,k)$-set consensus as an \emph{affine} task \cite{gact} $T$ (a task posed as a subcomplex of a subdivided simplex). 
With this task every processors can determine $k$-set consensus value for and combination of $m$ processors out of the $n \choose m$ possible. 
Then we approximate $T$ (using the colored simplicial approximation theorem \cite{HS,BG97}) by a subset of IIS runs.
This is done by completing $T$ to a subdivided simplex, approximating it, and choosing runs that land in $T$. 

\item Show that with $(m,k)$ set consensus $j$ processors can implement $k \lfloor j/m \rfloor + min(k,j \mod m)$-set consensus. Use this in conjunction of the simulation of SWMR
on IIS \cite{BG97,SERGIO-SIM}, the Generalized State-Machine Replication (GSMR) method \cite{SM}, and the RC simulation \cite{RC}, to run programs written for $MK$ to run on $\cal{RMK}$. 
The GSMR converts $MK$ to $n$ read-write threads 
progressing with maximum asynchrony of  $k \lfloor j/m \rfloor + min(k,j \mod m)$, where $j$ is the cardinality of the set of the active processors, processor that arrived with an input but did not obtain an output as yet. The $n$ fixed threads work as a BG \cite{BG,BGSERGIO} processors. The synchrony among them, as shown in the RC simulation, allows to do away with the Extended-BG simulation \cite{EBG} used for colorful tasks. 
Finally, the use of $(mk)$-set consensus in GSMR is replaced by the fact that we run in $\cal{RMK}$, while GSMR has iterated structure too; it evolves in rounds that use
``fresh'' variables in each round. This allows to simulate a GSMR round with a fixed number of iterations of $\cal{RMK}$. This extends the equivalence between IIS and and SWMR-SM from wait-free \cite{BG97,SERGIO-SIM}, 0-1 tasks \cite{01,SERGIO-SIM}, and $t$-resiliency \cite{FAST}, to consensus objects.

\item From the theory above it is clear that programs that call on $(m,k)$-objects can be compiled to run without the objects over $\cal{RMK}$. Can one write program ``directly'' with $\cal{RMK}$ in mind without ``cheating'' and using the compiler? We speculate that when the question is precisely formulated the answer will be positive. We show two rather simple example of algorithms that do not ``cheat.''
\end{enumerate} 

Since the quest to make the paper self contained was judged as hopeless, the traditional Model section section is abbreviate from \cite{gact} and put in the appendix.
In the following section we elaborate on the three points: 1. The idea behind getting $\cal{RMK}$ from IIS with $(m,k)$-objects, 2. How to compile MK programs to run on $\cal{RMK}$, and 3. Examples of direct design for $\cal{RMK}$. The next section goes into the details of the idea presented in 1 above, and the next section gives the technical detail of how to get the cumulative set-consensus power of $(m,k)$-objects. Finally the obligatory Conclusions.

\section{ Elaboration on the technical ideas} 

\subsection{A subset of $n$-processors IIS runs that solves $(m,k)$ set consensus among any $m$ processors, $n \geq m$, and is implementable by IIS with $(m,k)$-objects}

The (soft-wired \cite{OBJECTS}) object of  $(m,k)$-set consensus can be invoked by $m$ processors and to each invocation it returns one of the inputs provided to it by some invocation, such that at most $k$ distinct invocation  values are returned in total \cite{chaudhuri}.

It is known what is a $\cal{RMK}$ for $(2,1)$-objects. It is the iteration of the task IS \cite{IS} where in each iteration a processor returns a unique snapshot \cite{Common2} (The Distinct IS task). How do we do it in general for $(m,k)$-objects? In \cite{INSTANCY}, $k$-processors consensus is established as distinct IS task in which in addition a processor ``knows'' the $k-1$ ids preceding it
in the implicit order of processors of the task. But this author does not know how to translate ``know'' into IIS subset.

This is the main thrust of the paper. It builds on an ideas from \cite{2AD} and \cite{gact}. The paper in \cite{2AD} presents an elementary proof that the protocol complex of the read-write wait-free system constitutes a subdivided simplex. It does this by considering a sequence of pairwise interaction by processors, one pair per round. After all pairs are scheduled in distinct rounds, we obtain a composite chunk of rounds that now repeats ad infinitum.

The one-time pairwise interaction between two processors  in \cite{2AD} results in three possible outcomes for the pair, outcome 1, 2, and 3. Outcome 1 and 2 share a processor that has the same local view in both, and similarly outcomes 2 and 3. If we wanted to capture the system of read-write wait-free equipped with 2-processors consensus, we just do as above eliminating outcome 2. We will get a sub-complex of the subdivided simplex that we have obtained in the wait-free case. This simple observation drives the rest of the subsection.

We generalize the $(2,1)$ idea above to $(m,k)$, to get to built an affine complex  that solves $T$.

Finally, we do want to capture runs as a subset of a single model, which we chose to be the IIS. For that we show how
to replace a composite chunk of the rounds we described with an equivalent chunk from IIS. We use the idea from \cite{gact} that every \emph{affine} task, that is, a task that is a sub complex of a subdivided simplex, or for that matter any sub-complex of a chromatic subdivided simplex, can be equated with a set of IIS run. The composite chunk we have produced is an affine task $T$.

To capture any affine task as subset of IIS runs, we notice every colored subdivided simplex can be approximated by enough iteration of IIS \cite{HS,BG97} we complete the affine task to colored subdivided simplex, approximate it, and chose the prefixes of runs that fall into $T$.

The iterations of this chunk ad infinitum is $\cal{RMK}$

\subsection{$\cal{RMK}$ can run $MK$ Programs }

Let $\Pi$ be a program in $MK$, i.e. threads of reads and write to SWMR shared-memory, with the threads invoking $(m,k)$-set consensus objects. How do we execute $\Pi$
in $\cal{RMK}$? Processors in $\cal{RMK}$ can solve $(m,k)$-set consensus, but how do they coordinate local states calling on a copy of $(m,k)$ to know what is the outcome since
the virtual call to $(m,k)$ in $\cal{RMK}$ happens in a single round, while when the call to $(m,k)$ in $\Pi$ takes place by the processors simulating $\Pi$, at different rounds of the simulation?
Conceptually, the answer is simple: Just the first round in which the object is invoked matters. Latter processor will adopt a value from first round processors.
The implementation of this simple idea, unfortunately, requires heavy machinary.

We draw on 3 simulations: Simulating SWMR-SM on IIS \cite{BG97,SERGIO-SIM}, simulating free-for-all execution that builds on the replicated multi-state-machine in \cite{SM} (GSMR), and finally drawing on the companion submission called RC simulation that replaces the EBG simulation \cite{EBG}, by considering constant number $n$ of BG simulators but increasing and reducing their concurrency.  

We elaborate on each of these simulations in turn. The view from 20,000 feet is as follows: Processors run in $\cal{RMK}$ and drive a GSMR system. The role of a processor is to get its input into the simulation so that its thread can be executed. It is ignored (simulated as departed) once its thread has an output. The GSMR just gives steps to $n$ BG simulators. The less processors there are or the higher the power of consensus they have the higher the synchrony of the BG simulators. The $n$ BG simulators through the RC mechanism execute $\Pi$ (we need RC since we do not run BG in the traditional way of wait-free simulators but as simulators with certain level of synchrony).

\subsubsection{Simulating SWMR-SM on IIS}

Our target machine in this paper is a subset of of runs of the iterated system IIS. At the first step we would like to run the GSMR replication system of \cite{SM} in IIS. The replication system was written for SWMR-SM, but luckily it has a round structure.
At the beginning of a round all processors invoke set-consensus and then communicate within a SWMR mechanism.
The crucial observation is that from round to round GSMR uses ``fresh'' variables. Thus the simulation of a round of GSMR takes a fixed number of rounds in IIS: A fixed number of rounds of $\cal{RMK}$ (the chunck) to get the set consensus required, and then a fixed number of rounds to simulate the GSMR communication in a round (posting proposals, doing Commit-Adopt, etc.). 

Thus, our run in $\cal{RMK}$ is an alternating fixed size chunks of solving set consensus followed by a fix size chunk of rounds to simulate the read-write round of GSMR, solving set-consensus, etc.

\subsubsection{GSMR as a Threads Execution Model}

The scheme proposed in the Concur paper \cite{SM}, shows how to generalize the single State-Machine approach to distributed computing \cite{Lamport-clock} using consensus, to the case of $(\infty,k)$-set consensus, in short $k$-set consensus.
The state-machine approach \cite{Lamport-clock} shows how using consensus processors can coordinate to replicate a linear order of proposed \emph{commands}.
The GSMR assumes processors want to place commands on $k$ distinct machines. It shows how using $k$-set consensus they can replicate putting commands on these $k$ machines,
with progress guarantee that at least the placing of commands on at least one machine will progress.

A trivial Corollary of the technique behind GSMR in \cite{SM} shows that with $k$-set consensus, $n$ processors can place commands on $n$ state-machines with guaranteed progress on $n-(k-1)$
machines.

Thus, with $k$-set consensus $n$ processors can simulate $n$ threads where in each round at least $n-(k-1)$ of the threads advance. Thus if $k$ is small relative to $n$, the scheme simulates
an execution with high level of synchrony. With consensus, the execution of all the $n$ threads will be synchronous!

The main innovation of this section is to consider the state-machines to be read-write threads.
Processors running GSMR read their local replica, which may lag, or be ahead of another replica.
Based on their local read, they propose this value as a command for all the next read steps of threads (the writes will be inferred from the value of the read).
All these possibly distinct read values are proposed. Any one of them decided for a thread is a valid read value
of the tread since the threads read asynchronously. Thus the idea is to use GSMR as a execution scheme for read-write threads,
similar to the logic of the GB simulation. 

But the threads of $\Pi$ we are given are not only of read-write threads, they also invoke $(m,k)$-objects.
The next idea is to replace each $(m,k)$-object in $\Pi$ with a BG safe-agreement (SA) task \cite{}.
A solution to an SA task (see appendix) is read-write with a \emph{await(condition)} statement.
The burden we have is to show that these await statements will still allow progress of at least one thread.
When that thread will output, $\cal{RMK}$ processor associated with it (brought its input) will be simulated as departed,
the synchrony of GSMR will increase, the concurrency of executing $\Pi$ will hopefully decrease, and this will allow another thread to progress. 

Who are these $n$ threads we simulate? We do not simulate directly the threads of $\Pi$. GSMR is built with a fixed set of threads (state-machines) in mind. 
The effective set-consensus $\cal{RMK}$ that drives GSMR provides is implicit rather than explicit as it depends on the number of virtual arrival and virtual departure of
processors. This will affect the number of state machines that will progress. To do away with this complication GSMR runs a fixed number of threads $n$ of processors
that behave like BG simulators: They run all over $\Pi$ with some rule determining which thread of $\Pi$ can advance as some SA's are ``waiting.''

We could run a variable number of threads according to getting a GSMR thread that can progress.
This variable number of threads is a problem, since when synchrony grows and the number of threads shrink, threads that were active before and are not
active now may interfere with an SA task. In the past this problem was solved by the Extended-BG simulation \cite{EBG}. Here we solve it in a more elegant 
way by fixing the number of BG simulators to $n$, and letting the synchrony change.
Running the BG simulation with partially synchronous BG simulators, something that have not been done before, is described in a companion submission
under the name RC-simulation. The RC-simulation changes the BG scheme by determining that an SA is \emph{blocked} \cite{BG} only after some delay
to let live simulators have a chance to terminate their execution of the core of the SA (all but the await statement).
Thus, the RC-simulation \cite{RC} is an elegant substitute to EBG \cite{EBG}

\subsubsection{The RC Simulation}

The crux of the RC simulation was explained above. In more detail, suppose we have $n$ BG simulators with at most one fault.  We know \cite{BG,BGLR} that any number
of processors with at most a single possible fault is effectively two processors. The original BG simulator converts the above to let BG simulator number 1 and 2 take steps and all the rest
``skip.'' At least one of BG simulators 1 and 2 will take step. To do this we need to do the number 2, i.e. that at most a single processor might fail. The RC simulation lets all take steps. All will go to some SA, one will finish the core of the SA first without knowing the outcome.
Should it proceed to another SA? May be all the processors are alive and synchronous but they have started the SA at different times. If it will proceed to another SA the concurrency of the
execution of $\Pi$ might grow un-necessarily (and say, in the case of solving Renaming \cite{AR} will require more space than necessary). The crux of the RC simulation is to show that
if the decision to proceed to the next SA is delayed enough (as a function of the number of the active SA's) then if a simulator does proceed it is accounted for by one simulator being too
slow. I.e. the execution was not completely synchronous.

In case of at most one faults, after two SA's are active, the delay guarantees that at least one SA of the two will terminate. 

\subsubsection{Executing Threads that Invoke $(m,k)$-objects}
Now that we have reduced the execution to executing by BG processor in an RC simulation we use what we prove later that with 
$(m,k)$-set consensus $j$ processors can solve $k \lfloor j/m \rfloor + min(k,j \mod m)$-set consensus. This will be the effective number of BG processors in the RC simulation. Since the SA for every $(m,k)$-set consensus needs at least $k+1$ simulators to be siting in a middle of the safe agreement code, we get that at least one BG simulators can find a thread to execute.

\subsection{New-Line of Algorithms Design}

As mentioned in the subsection above, with $(m,k)$-set consensus $j$ processors can solve cumulative set consensus: $k \lfloor j/m \rfloor + min(k,j \mod m)$-set consensus. Take $m=2$, $k=1$. The object is now 2-processors consensus.

The most elementary task solvable by $n$ processors is Test-and-Set (TST). In TST one of the participants outputs ``win'' while the other output ``lose.'' What if at any point in the execution all we have is the cumulative set-consensus power of the 2-processors consensus? Can we do TST? Of course we can, as we can take any TST implementation and run it through the compiler we described. But then we replace objects with safe agreement etc. Can we do it directly? Of course this question is not formalized, but we will rely on Supreme-Court judge Potter Stewart saying: ``I know it when I see it''

It is elementary for TST. Processors do cumulative set consensus and write the id they obtained in shared-memory. A processor that afterwards does not see its id written outputs ``lose'' and depart (virtually). A processor that arrives late and sees any id written, outputs ``lose.'' Continue inductively with processors that saw their id written in shared-memory. Notice that this solution is linearazible \cite{HW}.

All algorithms in Common2 \cite{Common2,SWAP} ``TSTs everything that moves.'' It feels like the use of TST requires different mind-set than wait-free. Indeed, the group involved in Common2 over the years \cite{Common2,SWAP} seem to be the same ``old-hands.'' People who developed intuition in the use of TST. Lets therefore take the second most elementary
task solvable by 2-processors consensus: Tight-Renaming \cite{TIGHT}.

In Tight-Renaming each of $k$ participating processors outputs a unique integer in the range 1 to $k$. The standard TSTed way to solve it is for processors to TST the integers 1, 2,$\ldots$ in order with the winner outputting the integer it won. Can we accomplish the same using only the cumulative set-consensus  power of 2-processors consensus.

If the number of arrivals is $2k$ or $2k-1$ the cumulative set consensus power will narrow it to at most $k$. These at most $k$ processors can now solve Adaptive-Renaming \cite{AR} in the available range of at least 1 to $2k-1$. When one processor outputs it writes its output in shared memory and depart. The rest of the processors continue inductively using the integers that were not claimed by being written to shared memory. This solution is not linearizable.

This solution to tight-renaming illuminates how the ``wait-free logic'' of Adaptive-Renaming \cite{AR} spills over to the same problem type, when 2-processors consensus is available. It will be interesting to push the wait-free-logic to Fetch-and-Add and SWAP \cite{Common2,SWAP}. More importantly it will be interesting to formalize the question and actually prove it can always be done.

\section{Constructing $\cal{RMK}$}

We first present the task $(m,k)$-set consensus among any $m$ processors as an \emph{affine} complex $C(n,m,k)$, where $n$ denotes the number of processors. An Affine complex is a subcomplex of a chromatic finitely-subdivided simplex $A$.
To create $C(n,m,k)$ we first describe how we create $C(m,m,k)$, i.e. what is the subcomplex we talk about when the number of processors $n=m$. The complex   $C(m,m,k)$ is a subcomple
of $Chr^2 ({\bf s^{m-1}})$, the second standard chromatic subdivision of the $m-1$-dimensional simplex ${\bf s^{m-1}}$. To get $C(m,m,k)$ we purge from $Chr^2 ({\bf s^{m-1}})$ all the simplexes that are not part 
of an elementary $m-1$-dimensional simplex that has at least one vertex on a face of $Chr^2 ({\bf s^{m-1}})$ of dimension $k-1$. I.e. we hallow out $Chr^2 ({\bf s^{m-1}})$ of all simplexes that do not touch a $k-1$-dimensional face.
The observation we make leaving out the proof (as its straight forward but messy) is that any of the remaining simplexes every vertex not on a $k-1$-dimensional face has no two vertices on it link that touch two distinct  faces of dimension $k-1$ or less (this will not be true for the first subdivision). Thus every remaining simplex ``identifies'' exactly a single smallest face
of dimension $k-1$ or less.

\begin{lemma}
Consider the runs that corresponds to $C(m,m,k)$, then in a model of these runs we can solve $k$-set consensus among $m$ processors, and $C(m,m,k)$ when considered as a task is solvable in a read-write wait-free $m$ processors SWMR memory with access to $(m,k)$-set consensus.
\end{lemma}

\begin{proof}\\
$\Rightarrow$\\
For every run in a simplex of of $C(m,m,k)$ after two Immediate Snapshots a processor obtains a vertex of its color in $C(m,m,k)$. It then returns an id from the smallest cardinality face of $Chr^2 ({\bf s^{m-1}})$ of all the simplexes which contains it. By the unproven property we skipped the cardinality of the output set is $k$ or less.\\
$\Leftarrow$\\
Processor use $(m,k)$ set consensus to determine at most $k$ corners (0-dimensional faces) of $Chr^2 ({\bf s^{m-1}})$.  They then execute convergence \cite{HS,BG97} algorithm from these corners to return a simplex of $C(m,m,k)$.
\end{proof}

To construct $C(n,m,k)$, we consider all the $n \choose m$ combinations of $m$ processors in some order, $comb_1, \ldots, comb_{n \choose m}$.  We take the $n-1$-dimensional simplex ${\bf s^{n-1}}$ and take
the face that corresponds to $comb_1$. We subdivide it according to $C(m,m,k)$ and cone-off this subdivision with the rest of the vertices not in $comb_1$. Now we got a complex which is a subcomplex of a colored 
subdivided simplex. We take all the faces of elementary simplexes in the subdivision that correspond to $comb_2$. We subdivide each such face according to $C(m,m,k)$ and then in each simplex cone this subdivision off with the rest
of the vertices. We continue this for all combinations to get $C(n,m,k)$. 

$C(n,m,k)$ can be completed to a colored subdivision of ${\bf s^{n-1}}$. The completion viewed as a wait-free solvable task there exist $q$ such that $Chr^q  {\bf s^{n-1}}$ approximates the task.
We now consider all the simplexes of $Chr^q  {\bf s^{n-1}}$ that land in $C(n,m,k)$, to be the first $q$ rounds of runs in $\cal{RMK}$, denoted $\cal{RMK}_q$.

\begin{lemma}
Consider the runs that correspond to $\cal{RMK}_q$ , then in the model of these runs we can solve $(m,k)$-set consensus among all $n$ processors, and $\cal{RMK}_q$ when considered as a task is solvable  in a read-write wait-free $n$ processors SWMR memory with access to $(m,k)$-set consensus.
\end{lemma}
\begin{proof}
Since each simplex of $\cal{RMK}_q$ resides in a simplex of  $C(n,m,k)$ we can just identify back the process of the subdivision we described above and in turn each processor can answer its choice of value for each $m$ combination it belongs to. In the opposite direction our construction of $\cal{RMK}_q$, just used $(m,k)$-objects and read-write.
\end{proof} 

To create $\cal{RMK}$ we now iterate $\cal{RMK}_q$ ad-infinitum.

\ignore{
It is not too difficult to see that vertices of simplexes in $C(n,m,k)$ output atomic snapshots \cite{ATOMIC}. Each $p_i$ outputs a snapshot $S_i$ which is the carrier of the face of ${\bf s^{n-1}}$ that carries the vertex of its color, and the carriers
of simplexes of $Chr^k  {\bf s^{n-1}}$ are related as atomic snapshots. Moreover, obviously, by the logic of our process of creating $C(n,m,k)$, in none of the $n \choose m$ $(m,k)$-set consensus output $p_i$ returns it can return
an id of a processor $p_j$ such that $p_j \not \in S_i$.\\
}

\section{The cumulative set-consensus power of $j$ processors using $(m,k)$-objects}

\begin{theorem}
Given $(m,k)$-set consensus $j$ out of $n$ processors can implement $k \lfloor j/m \rfloor + min(k,j \mod m)$-set consensus.
\end{theorem} 
\begin{proof}


Consider we knew who the $j$ processors are.
W.l.o.g. \cite{OBJECTS} we assume soft-wired objects, i.e. the software controls that at most $m$ processors will invoke the $(m,k)$-set consensus.  To implement  $k \lfloor j/m \rfloor + min(k,j \mod m)$-set consensus processors rank themselves. We then arrange the objects in order, the first object covers the lowest $m$ ranked processors etc., and a processor invokes the object that covers the range of ranks that include its rank.

To ``know'' $j$, processors march through \emph{layers} 1 up to $n$. A processor starts at layer 1 with its id as input id.
Inductively it writes at layer $i$ all the ids it encountered at layer $i-1$, and takes a snapshot of all the ids written. If the number of distinct ids written is $i$,  it invokes the appropriate object, at layer $i$, with the inductive output id from layer $i-1$, obtains an input id, and writes the input id in shared memory. It then looks back at the number of arrivals to the layer. If it is still $i$, it departs with its input id. If it is larger than $i$, it continues with the input id it got from the object, and all the ids it encountered, to layer $i+1$.

On the other hand, if at layer $i$ the number of ids it observed is larger than $i$ it continues to layer $i+1$, with either its input id to layer $i$, in case it did not see an input id written at layer $i$ , else, it adopts an inout id written at layer $i$, as its input id to layer $i+1$.

The algorithm appears in figure \bf{Algorithm 1}. It works since if any processor returns it has seen the registration cardinality unchanged. Correspondingly processors that comes later will adopt a value from that layer. A concurrent processor that failed to see the registration cardinality unchanged obviously has a value from this layer. Since we go from 1 to $n$, some processor must return sometime. Notice the invariant that the value $IdSeen$ written to a layer is of cardinality greater equal to the layer index.
\end{proof}
\begin{algorithm}[t]

{\small
Shared Array $C1[1\ldots n,1 \ldots n]$ initialized to $\emptyset$\;
Shared Array $C2[1\ldots n,1 \ldots n]$ initialized to $\emptyset$\;
Local $IdSeen$ set of processors id, $InId$ input id, both initialized to $\{MyId\}$ and $MyId$, respectively\;
\
\
\\
%
%
%
%
       \For{ $j=1$ to $n$}{
      {$C[j,1]:=IdSeen$\; }
           {$Snap:= \cup_l~C1[j,l]$\;}
                            \eIf {$|Snap|=j$}{ { $InId:=$ Invoke with $InId$ the object according to $MyId$ rank in $Snap$\;}   
                {$C2[j,i]:=InId$\;}
                {$IdSeen=\cup_l~C1[j,l]$\;}
               {\bf{If} $|IdSeen|=j$ then return $MyId$\;}}
               {\bf{If} $\cup_l~C2[j,l] \not = \emptyset$ then $InId:=$ element of $\cup_l~C2[j,l]$\;}
               
              
               
        }
}

\caption{ Extracting cumulative power of soft-wired objects.}
\label{alg:lin}
\end{algorithm}

\subsection{ Implementing the Cumulative set-consensus power in $\cal{RMK}$}

Notice that the algorithm above is layered. I.e. each layer uses ``fresh'' objects, be it set consesnsus objects or read-write registers. This means that every layer can be simulated \cite{BG97,SERGIO-SIM} in $\cal{RMK}$ in some fixed number of iteration:
A processor that moved from layer $i$ to layer $i+1$ does not `interfere'' any more with writes at layer $i$. To see that the $(m.k)$-objects do not need to be persistent objects we make a further observation.   Each layer can be further partitioned into three phases: 
\begin{enumerate} 
\item The first phase is a read-write phase in which a processors writes the ids it encountered so far and takes a snapshot,
\item At the second phase a processor invokes an $(m,k)$-objects with it id, and
\item The third phase is again read-write phase in which a processor writes the id returned to it by the $(m,k)$-object, and look back at the ids now written in the first phase.
\end{enumerate}

Each of the three phases can be bounded a priori by some constant number of iterations. Thus the boundary of the middle phase is well defined and that is where processors in $\cal{RMK}$ simulate their invocations of the $(m,k)$-set consensus objects. Notice that in this phase processors in $\cal{RMK}$ in fact implement the soft-wired objects from hard-wired objects.

\section{Conclusions}
We have shown the existence and gave a constructive algorithm for the sub-IIS model that corresponds to any set-consensus objects. We remark in passing that generalizing this to any combination of such objects is straight forward. This adds evidence to support the quest of having \emph{canonical} distributed model, namely,  some subset of IIS runs. Supporting that any distributed computing system can be captured by a subset of IIS.

We showed further evidence to another Thesis, equating wait-free SWMR-SM equipped with any task, to the IIS model equipped with same. Namely, it was shown for 0-1 family of tasks, and now we enlarged to to consensus tasks.  Surprisingly, it is still unknown what is the sub-IIS model for the 0-1 family of tasks. 

Finally, and probably the least foundational but the most sexy part of the paper is the subsection that shows a new possibility of design of algorithms when consensus objects are available. In fact, we plan in the future of reproducing all the algorithms in Common2 \cite{Common2}, in that spirit.

\newpage

\appendix
\noindent Appendix
\\
\\
\\
\ignore{
The ramification of the trasformation of teh consensus objects to runs is that instead of writing a program where the program calls on teh objects one can write a straight line read-write program keeping in mind it will be executed with runs in the set built by the trasformation.

We show simple proof of concept example of teh ``new-line of algorithms design'' this opens by solving tight-adaptive-renaming \cite{} when $(2,1)$-set consensus objects are available.

Alternatively we show how to ``compile'' a program that calls on the objects to be executed by the run-model without access to the real objects.

Thus, it shows that at least for consensus objects, objects can be used in a generic standard manner.
Another ramification is that any program writen for a model where the objects are presitent, can be run in a model where objects instances are ephimeral - an object just leaves in a round. It is not available before and is not available after.
Such a replacement have been show in \cite{} for the 0-1 famility of tasks \cite{}.
} 
\section{Sub-IIS models}
\label{sec:models}

In this section, we describe our perspective on the \emph{Iterated
  Immediate Snapshot} (IIS) model~\cite{BG97} and give examples of
sub-IIS models. 

\subsection{The IIS model} \label{sec:IIS}
Our base model is the IIS. It consists of an infinite sequence of the IS tasks \cite{} $IS_1,IS_1, \ldots$. Processors  start by submitting their inputs to $IS_1$ and subsequently taking the output as the input to the next IS in the sequence. 

Let $\Runs$ be the set of runs in IIS. A processor is participating if it went through $IS_1$. A processor is live, if it  went ad-infinity.

Some processors might not be seen by other, or not seen infinitely often. This may allow to remove (a suffix) of their appearance in a run $r$ and still leave some processors unaware that we did this surgury. It is easy to see that this surgury has a well defined unique ``skeleton.'' The set of processors that are live in the skeleton $sk$ is called $\fast(sk)$. Since the skeleton is unique the set $\fast(r)$ will denote the fast set of the skeleton of $r$. In a run $r$, processor that are not in the $\fast(r)$ are in $\slow(r)$. 

\subsection{Examples of models} We define a {\em sub-IIS model} $M$ to be any subset of $\Runs$.

\begin{example}
\label{ex:wf}
The {\em wait-free} (or {\em completely asynchronous}) model $\WF$ is
the set $\Runs$ itself. The interpretation of $\WF$ is that anything
can happen (all sorts of step interleavings are allowed). 
\end {example}

\begin{example}
\label{ex:res}
For $t\leq n$, the {\em $t$-resilient model} $\Res_t$ consists of the runs $r \in \Runs$ such that $|\fast(r)| \geq n+1-t.$ This is the model in which at most $t$ processes are slow. 
\end{example}

\begin{example}
\label{ex:obstruction-free}
For $k\leq n+1$, the {\em $k$-obstruction-free model} $\OF_k$ consists
of all the runs $r$ in which no more than $k$ processes are fast,
i.e., $|\fast(r)| \leq k.$  This model was
previously discussed in \cite{Gaf08-concurrency},  following a suggestion of Guerraoui. \end{example}

\begin{example}
\label{ex:adv}
More generally, consider the {\em model with adversary $\A$} \cite{DFGT11}, which we denote by $M^{\adv}(\A)$. Here, $\A$ is any subset of the power set of $\{0, 1, \dots, n\}$.  We then define $M^{\adv}(\A)$ to consist of all runs $r$ such that $ \slow(r) \in \A$. 
\end{example}

\section{Topological definitions}
\label{sec:topdef}

We assume the reader is familiar with by now standard terminology used in Distributed Computing of Chromatic Complexes, Subdivided-Simplexes, etc

We denote by $\Chr^k \s$ the $k$'th iterated subdivision of the simplex $\s$, and by $|\Chr^k \s|$ we denote its some standard embedding in $R^n$. Since every simplex of $\Chr^k \s$ is a partition of $\Runs$ by prefixes if we continue this process to infinity we get that every point in the embedding of $\s$, $|\s|$, is a unique subset of $\Runs$. All runs at a point share the same skeleton.

\section{Tasks}
\label{sec:tasks}

\subsection{Definitions}
\label{sec:deftasks}

A {\em task} $T = (\I, \O, \Delta)$ on $n+1$ processes $\{p_0, \ldots,
p_n\}$ consist of two finite, pure $n$-dimensional chromatic complexes $\I$ and $\O$, together
with a chromatic multi-map $\Delta: \I \to 2^{\O}$. The {\em input complex} $\I$ specifies the possible input values, the {\em output complex} $\O$ specifies the possible output values, and $\Delta$ describes which output values are allowed for a given input. The colors specify to which process each input or output value corresponds.

A task is called {\em input-less} if the input complex is the standard
simplex $\s$, colored by the identity. 
Then each process starts with input only its own id.\footnote{Note that in the definition of a multi-map we allowed images to be empty. This is somewhat non-standard, as it means that processes in a task do not have to output. If one prefers to avoid that, for every task $T=(\I, \O, \Delta)$ we can construct a new, equivalent task $T^+ = (\I^+, \O^+, \Delta^+)$ as follows. We let $\I^+=\I$. The output complex $\O^+$ is obtained from $\O$ by adding extra vertices $v_0, \dots,v_n$ (with $v_i$ corresponding to ``no output'' for the process $i$); moreover, for each simplex $\sigma$ in $\O$, we add an $n$-simplex $\sigma^+$ in $\O^+$ by adjoining vertices $v_i$ for the colors $i$ not represented in $\sigma$. Finally, we let $\Delta^+(\tau) = (\Delta(\tau))^+$.}

\subsection{Affine tasks}
Many examples of input-less tasks can be constructed as follows. Let $L \subseteq \Chr^k \s$ be a pure $n$-dimensional subcomplex of the $k\th$ chromatic subdivision of $\s$, for some $k$. For each face $\t \subseteq \s$, the intersection $L \cap \Chr^k \t$ is a subcomplex of $\Chr^k \s$; we assume that this subcomplex is pure of the same dimension as $\t$ (and possibly empty).  

We define an input-less task $(\s, L, \Delta)$ by setting 
$\Delta(\t) = L \cap \Chr^k \t$ for any face $\t \subseteq \s$.  Tasks constructed like this are called {\em affine}. To depict an affine task, we can simply draw the corresponding complex $L$.

By abuse of notation, we will usually write $L$ for the affine task $(\s, L, \Delta)$. We chose the name {\em affine} because if we have a task $L$ as above, the geometric realizations of the simplices of $L$ can be depicted as lying on affine subspaces of $\R^n$. Similar terminology appears in algebraic geometry, where one talks about affine varieties.

\subsection{Task Solvability}
In a sub-IIS model, informally, a task $T = (\I, \O, \Delta)$ is solvable in $M$ if for all runs $r \in M$, the infinitely participating processes output, and their output is a subsimplex of the allowed outputs for the participating processes. An output is the result of a {\em protocol}. For us, when dealing with solvability rather than complexity, a protocol is just a partial map from views to outputs. Thus, requiring an infinitely participating process to output means requiring that eventually it will have a view that is mapped by the protocol to an output value.

We define the set $\V= \V(\I)$ to consist of all possible $\view(p_i, \omega, k)$ in all runs $r \in \Runs$, for all processes $p_i$, simplices $\omega \in \I$, and integers $k \geq 0$. Formally, a protocol $\Pi$ for the task $T$ is a map from a subset of $\V$ to the set of vertices in the output complex $\O$.

\begin{definition}
A task $T = (\I, \O, \Delta)$ is {\em solvable} in a sub-IIS model $M$ if there exists a protocol $\Pi$ for $T$ such that for all $r \in M$ (with $r=S_1, S_2, \dots$ as before):
\begin{enumerate}
\item For each $p_i$, and for each $n$-dimensional simplex $\omega \in \I$, there exist $k_0$ and a vertex $v$ of $\O$ colored $i$, such that:
\begin{itemize}
\item For all $k <k_0$,  $\view(p_i, \omega, k) \notin \textit{domain}(\Pi)$;
\item For all $k \geq k_0$ such that ${p_i} \in S_k$ exists, we have $\Pi(\view({p_i}, \omega, k))=v$.
\end{itemize}

\noindent (This condition is satisfied vacuously if $p_i$ is not infinitely participating, because we can find $k_0$ such that $p_i$ did not take $k_0$ steps in $r$, so $p_i \not \in S_k$ for $k \geq k_0$.)
 \vskip5pt
 
\item For all $k$,  $\{\Pi(\view(p_i, \omega, k)) \mid \view(p_i,\omega, k)\in \textit{domain}(\Pi)$\} is a sub-simplex of a simplex in $\Delta\bigl (\omega \cap \chi^{-1}(\part(r)) \bigr)$.
\end{enumerate}
\end{definition}

In every run $r\in M$, condition (1) above requires every infinitely
participating to eventually produce an output, and condition (2)
requires the produced output to respect the task specification
$\Delta$ given the inputs of participating processes.   

\subsection{Safe-Agreement (SA) Task}

We present here safe-agreement as a task. In the literature safe-agreement is specified operationally \cite{BGSERGIO}.


\begin{enumerate}
\item Processor $p \in P$ outputs $\bot$ or its input $v_p$,
\item At least on processor $p \in P$ does not output $\bot$,
\item All processors that do not output $\bot$ output the same $v_q, ~q \in P$.
\end{enumerate}

The SA task can be solved wait-free \cite{BG}. A wait-free solution to SA is a {\em SA-module}.
A SA-module also asks processors to post their output in Shared-Memory. The implication is
that either $p \in P$ that terminated knows an non-$\bot$ output $q$ to SA, or if not, there is at least one processor $p' \in P$
that has not terminated the SA-module. This idea that either processors know the value of the election
in the SA, or otherwise one processor $p'$ invoked SA but has not returned (or returned but did not write its return,
which is always the next to do after a return) from SA, and consequently blocked from executing.
We call a processor that invoked SA, did not post an output, and an output $q$ is not available, a processor
that is blocking or stalling the SA.


\begin{thebibliography}{99}

\bibitem{FAST}
Zohir Bouzid, Eli Gafni, Petr Kuznetsov: Live Equals Fast in Iterated Models. CoRR abs/1402.2446 (2014)


\bibitem{HW}
Maurice Herlihy, Jeannette M. Wing: Linearizability: A Correctness Condition for Concurrent Objects. ACM Trans. Program. Lang. Syst. 12(3): 463-492 (1990)

\bibitem{RC}
Pierre Fraigniaud, Eli Gafni, Sergio Rajsbaum, Mathieu Roy: Automatically adjusting concurrency to the level of synchrony. Submitted to DISC 2014.

\bibitem{BG97}
Elizabeth Borowsky, Eli Gafni: A Simple Algorithmically Reasoned Characterization of Wait-Free Computations (Extended Abstract). PODC 1997: 189-198.


\bibitem{TIGHT}
Yehuda Afek, Eli Gafni, Opher Lieber: Tight Group Renaming on Groups of Size g Is Equivalent to g-Consensus. DISC 2009: 111-126.

\bibitem{SWAP}
Yehuda Afek, Adam Morrison, Guy Wertheim: From bounded to unbounded concurrency objects and back. PODC 2011: 119-128.

\bibitem{Common2}
Yehuda Afek, Eytan Weisberger, Hanan Weisman: A Completeness Theorem for a Class of Synchronization Objects (Extended Abstract). PODC 1993: 159-170.


\bibitem{Lamport-clock}
Leslie Lamport: Time, Clocks, and the Ordering of Events in a Distributed System. Commun. ACM 21(7): 558-565 (1978).

\bibitem{2AD}
Yehuda Afek, Eli Gafni: Asynchrony from Synchrony. ICDCN 2013: 225-239.

\bibitem{INSTANCY}
Yehuda Afek, Eytan Weisberger: The Instancy of Snapshots and Commuting Objects. J. Algorithms 30(1): 68-105 (1999).

\bibitem{SERGIO-SIM}
Eli Gafni, Sergio Rajsbaum: Distributed Programming with Tasks. OPODIS 2010: 205-218.

\bibitem{gact}
Eli Gafni, Petr Kuznetsov, Ciprian Manolescu: A generalized asynchronous computability theorem. CoRR abs/1304.1220 (2013). To appear in PODC2014.

\bibitem{chaudhuri}
Soma Chaudhuri: Agreement is Harder than Consensus: Set Consensus Problems in Totally Asynchronous Systems. PODC 1990: 311-324.

\bibitem{HS}
Maurice Herlihy, Nir Shavit: The topological structure of asynchronous computability. J. ACM 46(6): 858-923 (1999).



\bibitem{ATOMIC}
Afek Y., H. Attiya, Dolev D., Gafni E., Merrit M. and Shavit N.,
Atomic Snapshots of Shared Memory.
{\em Proc. 9th ACM Symposium on Principles of Distributed Computing (PODC'90)},
ACM Press, pp. 1--13, 1990.

\bibitem{AR}
Hagit Attiya, Amotz Bar-Noy, Danny Dolev, David Peleg, R�diger Reischuk: Renaming in an Asynchronous Environment J. ACM 37(3): 524-548 (1990)

\bibitem{IS}
Elizabeth Borowsky, Eli Gafni: Immediate Atomic Snapshots and Fast Renaming (Extended Abstract). PODC 1993: 41-51.





\bibitem{SM}
Eli Gafni, Rachid Guerraoui: Generalized Universality. CONCUR 2011: 17-27.



\bibitem{BG}
Elizabeth Borowsky, Eli Gafni: Generalized FLP impossibility result for t-resilient asynchronous computations. STOC 
1993: 91-100.




\bibitem{BGSERGIO}
Elizabeth Borowsky, Eli Gafni, Nancy~A. Lynch, and Sergio Rajsbaum:
The {BG} distributed simulation algorithm.
{\em Distributed Computing}, 14(3):127--146, 2001.


\bibitem{EBG}
Eli Gafni:
The extended {BG}-simulation and the characterization of
t-resiliency.
In {\em STOC}, pages 85--92, 2009.


\bibitem{OBJECTS}
Elizabeth Borowsky, Eli Gafni, Yehuda Afek: Consensus Power Makes (Some) Sense! (Extended Abstract). PODC 1994: 363-372.



\bibitem{01}
Eli Gafni: The 0-1-Exclusion Families of Tasks. OPODIS 2008: 246-258.

\end{thebibliography}
\end{document}